\documentclass[11pt]{article}

\usepackage{fullpage}
\usepackage{hyperref}
\usepackage{graphicx}

\headheight \baselineskip
\newcommand{\shortv}[1]{}

\usepackage{tonsmod}
\newcommand{\citeyear}{\cite}

\newcommand{\commentout}[1]{}

\usepackage{verbatim}
\usepackage{amssymb}
\usepackage{amsmath}
\usepackage[noend]{algpseudocode}
\usepackage{algorithm}
\usepackage{amsthm}
\usepackage{enumitem}

\usepackage{xspace}

\usepackage{amsthm}
\usepackage{thmtools, thm-restate}
\declaretheorem{theorem}
\declaretheorem{definition}

\declaretheorem{property}

\newcommand{\A}{\mathcal{A}}

\begin{document}
%

\title{A Distributed Auctioneer for Resource Allocation in Decentralized Systems}


\author{Amin M. Khan\\
Universitat Polit\`{e}cnica de Catalunya, Barcelona\\
mkhan@ac.upc.edu
\and Xavier Vila\c{c}a\\
INESC-ID Lisboa\\
xvilaca@gsd.inesc-id.pt
\and Lu\'{i}s Rodrigues\\
ler@tecnico.ulisboa.pt
\and Felix Freitag\\
Universitat Polit\`{e}cnica de Catalunya, Barcelona\\
felix@ac.upc.edu
}

\maketitle



\begin{abstract}

In decentralized systems, nodes often need to coordinate to access shared resources in a fair manner. One approach to perform such arbitration is to rely on auction mechanisms. Although there is an extensive literature that studies auctions, most of these works assume the existence of a central, trusted auctioneer. Unfortunately, in fully decentralized systems, where the nodes that need to cooperate operate under separate spheres of control, such central trusted entity may not exist. Notable examples of such decentralized systems include community networks, clouds of clouds, cooperative nano data centres, among others. In this paper, we make theoretical and practical contributions to distribute the role of the auctioneer. From the theoretical perspective, we propose a framework of distributed simulations of the auctioneer that are Nash equilibria resilient to coalitions and asynchrony. From the practical perspective, our protocols leverage the distributed nature of the simulations to parallelise the execution. We have implemented a prototype that instantiates the framework for bandwidth allocation in community networks, and evaluated it in a real distributed setting.

\end{abstract}



\section{Introduction}
\label{sec:introduction}

Most distributed systems have limited resources that need to be shared by many nodes. 
For instance, in a network there is limited bandwidth that needs to be allocated to multiple nodes. 
In cloud applications, virtual machines (VMs) need to be allocated to different cloud users. 
Resource allocation is, therefore, a key problem in distributed systems.

Distributed resource allocation is particularly challenging when nodes
operate under different spheres of control and may not be willing to cooperate. 
Namely, a resource allocation strategy that assumes that all nodes 
execute a given algorithm may break if nodes 
may extract benefits by deviating from the expected behaviour.
Many examples of this problem can be found in the literature. 
The works of~\cite{Lee2007} and~\cite{Xiao2013} illustrate how
a network user may attempt to monopolize the bandwidth utilization 
if it has the opportunity. 
There is evidence that programmers can instrument their code 
to get an unfair advantage of several Unix schedulers~\cite{Grosu2005}. 
In shared infrastructures, like Grid systems, participating
users try to maximise their own usage to the detriment of 
the others~\cite{Lai2004}.
Dynamic wireless spectrum allocation suffers from unfair manipulation~\cite{Zhou2008}.
Social cloud computing~\cite{Caton2014}, 
and cooperative computing systems like BitTorrent~\cite{Liu2010}
suffer when users act selfishly in consuming resources.

An approach that has emerged as a viable alternative for the problem
above is to use economic models to address resource allocation, in
particular by resorting to auction systems~\cite{Riley1981}. 
As a result, an extensive
literature exists on the use of several types of auctions to perform
resource allocation in distributed systems~\cite{Niyato2013, Lai2004, Waldspurger1992}. 
In particular, the advent of the cloud computing model, where many
clients may compete for the resources managed by one or more
providers, has spurred the usage of different auction mechanisms in a
variety of resource allocation proposals for the
cloud~\cite{Popa2012, Niu2012, Wang:12, 
Zhang:13, Shi:14, Zheng2014Star, Zhang:14, Zhang2015Truthful}.

In these approaches, users are modelled as non-cooperative rational
players who are willing to pay for using resources or get paid for
providing those resources. Specifically, users declare to an
auctioneer the preference for different allocations of resources, and
the auctioneer executes some auction mechanism to derive an allocation
between users and resources that maximises social welfare (preferences of users for the allocation),
and the payments to be performed or
received by each user. The aim is to obtain an allocation with a social welfare as close as
possible to the optimal while ensuring truthfulness from the
users, such that they do not have incentives to lie about their
bids. In addition to maximal social welfare and truthfulness, other guarantees may be provided,
including computational efficiency and budget balance (the
payments made by the users outweigh the payments received).

These works assume that the auctioneer is trusted.  Unfortunately,
this is an unreasonable assumption in many of today's fully
decentralized systems, where all nodes are either resource consumers,
resource providers, or both. 
In this case, there is no natural candidate that 
can be trusted by all other nodes to run the auction
algorithm faithfully, given that any node may extract some
benefit by perturbing the auction result. In some sense, all current
distributed systems that rely on auctions to perform resource
allocation are not fully decentralized, because they depend on a
unique central point of control, which is the trusted node that runs
the auction algorithm. 
This leads to the observation that there is
still a substantial gap that needs to be bridged to apply these
results in fully decentralized settings.

This paper bridges this gap by proposing a framework of distributed
protocols that allows multiple resource providers in a decentralized
distributed system to simulate the role of the auctioneer.  Such simulation
raises significant challenges both from the theoretical and practical
points of view. From the theoretical perspective, although there is a
vast literature of distributed fault-tolerant algorithms, with very
few exceptions (for instance,~\cite{Abraham:13}), these works do not
consider rational behaviour. From the practical perspective, the
distribution of the auctioneer may incur in additional overhead. Our
paper addresses both concerns. First, we prove that our distributed
simulations are sound from a game theoretical
perspective. Specifically, we show that the simulations are
$k$-resilient (ex post) equilibria~\cite{Abraham:13}, i.e., Nash
equilibria resilient to asynchrony and coalitions of providers of size
at most $k$. Second, our protocols leverage the distributed nature of
the resulting virtual trusted entity to parallelise the resource
provisioning algorithm, compensating for the additional costs imposed
by coordination. 

We have implemented a prototype of our protocols and evaluated the
resulting system in a real instance of a fully decentralized network,
namely an experimental testbed for community networks~\cite{Braem2013}.  
Different from the traditional
business-focused model applied by telecommunication operators, each
user in a community network is an owner of a portion of the total
infrastructure, which builds the mesh network.
Community networks are an interesting application scenario for
our results  because they lack a central point of control and make
extensive use of resource sharing. In particular, we consider the
concrete problem of bandwidth reservation on the gateways that connect
the community network to the Internet.  

Even though in our example we only focus on bandwidth
reservation, the fundamental problems at hand are general and emerge
every time shared resources need to be allocated to users in a set of
providers, such as for allocation of processing and memory resources
as virtual machines in public clouds~\cite{Zhang2015Truthful},
assignment of frequencies in secondary wireless spectrum
markets~\cite{Zhou2008}, and resource scheduling in grid and cloud
infrastructures\cite{Lai2004}.

In summary, the main contributions of this paper are the following:

\vspace{0.3mm}
- We propose a framework for devising distributed protocols executed
among a decentralized set of service providers that correctly simulate the
auctioneer in a family of resource allocation auctions.

\vspace{0.3mm}
- We show that every implementation of the framework is a $k$-resilient (ex post) equilibrium.
These implementations also tolerate users that send invalid bids.

\vspace{0.3mm}
- We show that it is possible to leverage the distributed nature of our framework to parallelise implementations,
mitigating scalability issues of purely centralised solutions.

\vspace{0.3mm}
- We implemented instances of the framework and reported the results from its deployment on the actual Guifi nodes.

\vspace{0.3mm} 

The rest of the paper is organised as follows.  In
Section~\ref{sec:related-work} we present related work.
Section~\ref{sec:model} provides the system model.  In
Section~\ref{sec:design}, we present the framework for simulating the
auctioneer and in Section~\ref{sec:instances} we illustrate its
applicability to the execution of two different auction mechanisms,
with different computational properties, in the context of community networks.  In
Section~\ref{sec:evaluation} we discuss results from the experiments.
Section~\ref{sec:conclusion} concludes the paper.


\section{Related Work}
\label{sec:related-work}

A vast literature has addressed the problem of allocating resources
between providers and users following an auction approach~\cite{Nisan2001, 
Zheng2014Star, Wang:12,Zhang:13, Popa2012, Niu2012,
Shi:14,  Zhang:14, Zhang2015Truthful}.
For instance, in~\cite{Zhang:14,Zhang2015Truthful}, the authors propose Vickrey–Clarke–Groves (VCG) mechanism~\cite{Nisan2001}
in auctions where only the users submit bids to the auctioneer,
achieving truthfulness and a tradeoff between maximal social welfare and computational efficiency;
Zheng et al.~\cite{Zheng2014Star} propose a variant of the McAfee mechanism 
to tackle a similar problem in a double auction (providers also submit bids),
where they achieve truthfulness and budget balance.
To the best of our knowledge, none of these works addressed these problems
in the absence of a trusted auctioneer.

The problem of simulating the behaviour of a trusted entity
in an environment with only rational players has been approached
in the literature of distributed systems\,\cite{Halpern:04,Abraham:06,Abraham:13,Afek:14}.
In~\cite{Halpern:04,Abraham:06}, the authors addressed the
particular problems of secret sharing and multiparty computation
assuming the existence of a trusted mediator, and then studied
conditions under which it is possible to simulate the mediator
through a distributed protocol. Abraham et al.\,\cite{Abraham:13}
devised $k$-resilient equilibria solutions for the problem of leader election.
Afek et al.\,\cite{Afek:14} proposed a building blocks approach
for devising distributed $k$-resilient implementations, and
used this approach in combination with ideas from~\cite{Abraham:13}
to address the problems of consensus and renaming.
None of these works devised distributed protocols for simulating 
the role of an auctioneer in an auction.


\section{System Model}
\label{sec:model}
We define the family of resource allocation auctions, the requirements of 
a distributed simulation of the auctioneer, and the Game Theoretical model
used to analyse simulations.

\subsection{Resource Allocation Auctions}
Consider a family of auctions with $m$ \emph{providers},
$n$ \emph{users}, and an \emph{auctioneer}.
Providers sell multiple resources with a limited capacity,
in exchange for payments in some currency. Users are willing to pay
to the providers in exchange for the allocation of a minimum amount 
of each resource in the provider.
The auctioneer defines an allocation between users and providers
that is \emph{feasible} (i.e., that does not exceed 
the capacity of each resource in any provider)
and defines the payments to be made/received by the users/providers, respectively.
Both users and providers attribute a utility to each allocation,
which is a function of the value given to the allocation and the payments made/received.
More precisely, each user $i$ has a valuation $v_i$ 
specifying how much $i$ is willing to pay for the allocation
of a unit of each resource in each provider;
$i$'s utility is the difference between the total value
attributed by $i$ to the allocation and the payments made by $i$.
On the other hand, the valuation $v_j$ of a provider $j$
specifies how much $j$ wants to be paid for allocating a unit of each resource;
$j$'s utility is the difference between the payments received by $j$
and the total value attributed by $j$ to the allocation.

We will analyse two types of auctions: \emph{standard} and \emph{double}
that differ only on who are the bidders (entities that submit bids).
In a standard auction, only the users are bidders.
Each user $i$ submits a bid $b_i$ to the auctioneer declaring $v_i$.
Then, the auctioneer executes an algorithm $\A$ that
returns a feasible allocation and the respective payments.
The algorithm $\A$ must satisfy three properties: (1) it must maximise in expectation 
the \emph{social welfare}, defined as the total value attributed by users to the allocation,
(2) it must achieve \emph{truthfulness in expectation}, i.e.,
no user may increase its expected utility by lying in its bid;
and (3) it must be \emph{computationally efficient}.
In a double auction, the auctioneer collects bids from both the users and the providers.
The social welfare is now the
difference between the total value of the users and the total 
value of the providers. In addition to the above three
properties, $\A$ should also satisfy
\emph{budget balance}, which is the property that the total
value paid by users covers the total payments made to the providers.
Unfortunately, it was shown that no algorithm can simultaneously satisfy
truthfulness in expectation, maximal social welfare, and budget balance\,\cite{Myerson:83}.
In practice, it is common to aim at a combination between
truthfulness in expectation and either one of the other two properties.

\subsection{Distributed Auctioneer Simulation}
In our setting, no single entity can be trusted with the role of the auctioneer,
since every entity may increase its utility by manipulating the execution of $\A$.
Specifically, a provider may devise a sub-optimal schedule
that provides him with a higher payment; similarly, a bidder
may manipulate the execution of $\A$ to decrease his payment.
We address this problem by simulating the role of the auctioneer through a distributed protocol.
The idea is to replicate the execution of $\A$ in multiple entities
and use cross-validation of the results of the redundant computations.
Providers are especially suited for this purpose, since they may be willing
to offer their resources for the execution of $\A$ in exchange for payments.
Therefore, we focus on distributed protocols executed among sets
of providers that deviate from the protocol only if they gain by doing so.

Now, we specify requirements for a correct simulation of the auctioneer.
Normally, the auctioneer collects a vector $\vec{b}$
of bids and executes $\A$ with input $\vec{b}$.
In a simulation, each provider $j$ must collect a vector $\vec{b}^j$
of bids sent to $j$ and use it as input of a distributed protocol
that simulates $\A$. This requires bidders to submit a bid
to all providers. We consider that bidders may adopt arbitrary behaviours
such as submitting different bids to different providers or not submitting a bid.
Nevertheless, we assume that every provider $j$
eventually collects a vector $\vec{b}^j$ to be used as input in the simulation,
containing a bid for every bidder $i$, and if $i$ is correct,
then $j$ receives the bid of $i$ prior to the simulation.
In practice, bidders are expected to submit their bids by some deadline;
if a bidder fails to do so or sends an invalid bid,
then the provider may use the special value $\bot$ instead.
We want a simulation of the auctioneer to simulate $\A$ on some 
input $\vec{b}$ that contains at least the bids sent by correct bidders,
regardless of the bids of remaining bidders.

More precisely, let $\A(x,\vec{p} \mid \vec{b})$ be the probability
of algorithm $\A$ outputting an allocation $x$ and vector of payments $\vec{p}$,
when executed on input $\vec{b}$ by a trusted auctioneer.
We denote by $b_i^j$ be the bid submitted by
bidder $i$ to provider $j$ in a simulation.
Let $\vec{b}^j$ be the vector of all bids sent to $j$.
If $i$ does not submit a valid bid to $j$, then we take $b_i^j$
to be a neutral bid (i.e., a bid that excludes $i$ from the auction).
In a simulation of the auctioneer, each provider $j$ inputs $\vec{b}^j$
and outputs a pair $(x,\vec{p})$ composed by an allocation $x$
and a vector of payments $\vec{p}$, or outputs a special value $\bot$
that signals the abortion of the simulation.
We say that the outcome is $(x,\vec{p})$ if all providers output this pair,
otherwise, the outcome is $\bot$.
We assume that an external mechanism guarantees that
(1) when the outcome is $\bot$, the auction is aborted,
and (2) when the outcome is $(x,\vec{p})$, the allocation $x$
is enforced and all entities perform or receive their respective payments.
We can now provide a precise definition of correct simulation.

\begin{definition}
A simulation is said to be correct if and only if,
for all vectors $(\vec{b}^j)_j$, the outcome is $(x,\vec{p})$ with
probability $\A(x,\vec{p} \mid \vec{b})$, where $\vec{b}$
only contains valid bids and, for all bidders $i$ such that 
$b_i^j = b_i'$ for every provider $j$, we have $b_i = b_i'$.
\end{definition}

\subsection{Game Theoretical Model}
We consider the model of extensive form games played in 
asynchronous systems proposed in~\cite{Abraham:13}.
There are $m > 1$ players corresponding to the providers 
of the resource allocation auction.
Providers may form coalitions of size at most $k$. 
We assume that each provider has a unique identifier,
known to every other provider. Time is divided into turns.
In each turn, some provider $j$ is chosen to move:
$j$ first receives messages previously sent to $j$,
performs some computation, and sends messages.
A schedule specifies which provider moves at each turn
and which messages it receives.
We assume that communication channels are reliable,
so every message sent is eventually delivered.
We focus on schedules that are \emph{fair} in the sense
that every provider $j$ is scheduled to move infinitely
often, such that, for all turns $t$, there exists a turn $t'>t$
when $j$ is scheduled to move. This is necessary to ensure progress.

Now, we want to define a notion of equilibrium for this setting.
For this, we need an exact definition of protocol and utility.
A protocol specifies, for each schedule, a probability distribution over
the computation performed by each provider $j$ and the messages sent at each turn where $j$ moves,
as a function of the history of messages sent and received in previous turns.
In addition, since we analyse protocols as modules with input and output values,
the protocols also specify the values used as input and output by each provider.
The utility of providers is a function of the outcome of the simulation:
if the outcome is $\bot$, then the utility is $0$, else the
utility is the difference between the payments received and the value of the allocation.
Given this, the utility of a user $i$ is also $0$ if the outcome is $\bot$,
or is the difference between the value of the allocation and the payments made.
The expected utility conditioned on the schedule is
computed according to the probability distribution over outcomes induced by the protocol.

For the definition of equilibrium, we consider the notion of $k$-resilient 
(ex post) equilibrium introduced in~\cite{Abraham:13}.
This notion is a refinement of Nash equilibrium that incorporates
collusion and asynchrony. Namely, collusion is modelled as sets $K$
of at most $k$ providers that coordinate on any joint protocol;
we require that no provider in $K$ can increase its expected utility
if providers in $K$ deviate from the specified protocol,
given that other providers do not deviate.
Asynchrony is modelled in an ex post way by assuming that providers 
are informed about the schedule when computing the expected utility.
This is the strongest requirement for this setting.

\begin{definition}
A protocol $P$ is a $k$-resilient (ex post) equilibrium if and only if
for all fair schedules and coalitions $K$ such that $|K| \le k$,
there is no provider in $K$ that increases its expected utility
when providers in $K$ follow a joint protocol $P' \ne P$,
given that providers not in $K$ follow $P$.
\end{definition}

An important aspect of this definition is that a $k$-resilient
equilibrium protocol $P$ satisfies the property that no provider increases
its expected utility by lying about its input, hence $P$
achieves \emph{truthfulness} regarding the inputs of the providers to a simulation.
Specifically, at every invocation, we say that provider $j$ has input $v$
if, by following $P$, $j$ is expected to input $v$, so truthfulness implies that $j$ 
does not input $v' \ne v$.
By the truthfulness of $\A$, every bidder $i$ maximises its expected utility
by sending $b_i = v_i$ to all providers, regardless of other bids.
This implies that to fulfil our goals it suffices to devise
protocols that are $k$-resilient and correctly simulate the auctioneer.


\section{The Distributed Auctioneer}
\label{sec:design}

We propose a framework for devising distributed 
protocols executed by the providers that correctly simulate the auctioneer.
The framework is sufficiently general to simulate different auctions.
To illustrate its applicability, we provide two implementations 
of the framework for standard and double 
bandwidth allocation auctions, respectively.
We describe the framework in two steps. First, we provide a general definition
where we do not specify the details about how to implement the simulation
of the algorithm $\A$. Then, we describe how to simulate $\A$
by leveraging parallelism to speed up its execution.

\subsection{General Framework}
The input of the framework at each provider $j$ is a vector $\vec{b}^j$
of bids submitted to $j$ and the output is either a pair $(x,\vec{p})$ containing
an allocation $x$ and a vector of payments $\vec{p}$ or the special value $\bot$.
As illustrated in Figure~\ref{fig:framework}, the framework chains the execution
of two building blocks: \emph{bid agreement} and \emph{allocator}.
Each provider $j$ inputs $\vec{b}^j$ to the bid agreement,
which outputs either a vector $\vec{b}$ or $\bot$.
In the former case, $j$ inputs $\vec{b}$ to the allocator.
If all providers follow the protocol, then the bid agreement ensures
that they all output some vector $\vec{b}$ containing all valid bids,
and the allocator ensures that they all output a pair $(x,\vec{p})$
with probability determined by $\A$.

\begin{figure}[tbp]
	\centering
	\includegraphics[width=0.5\columnwidth,keepaspectratio]{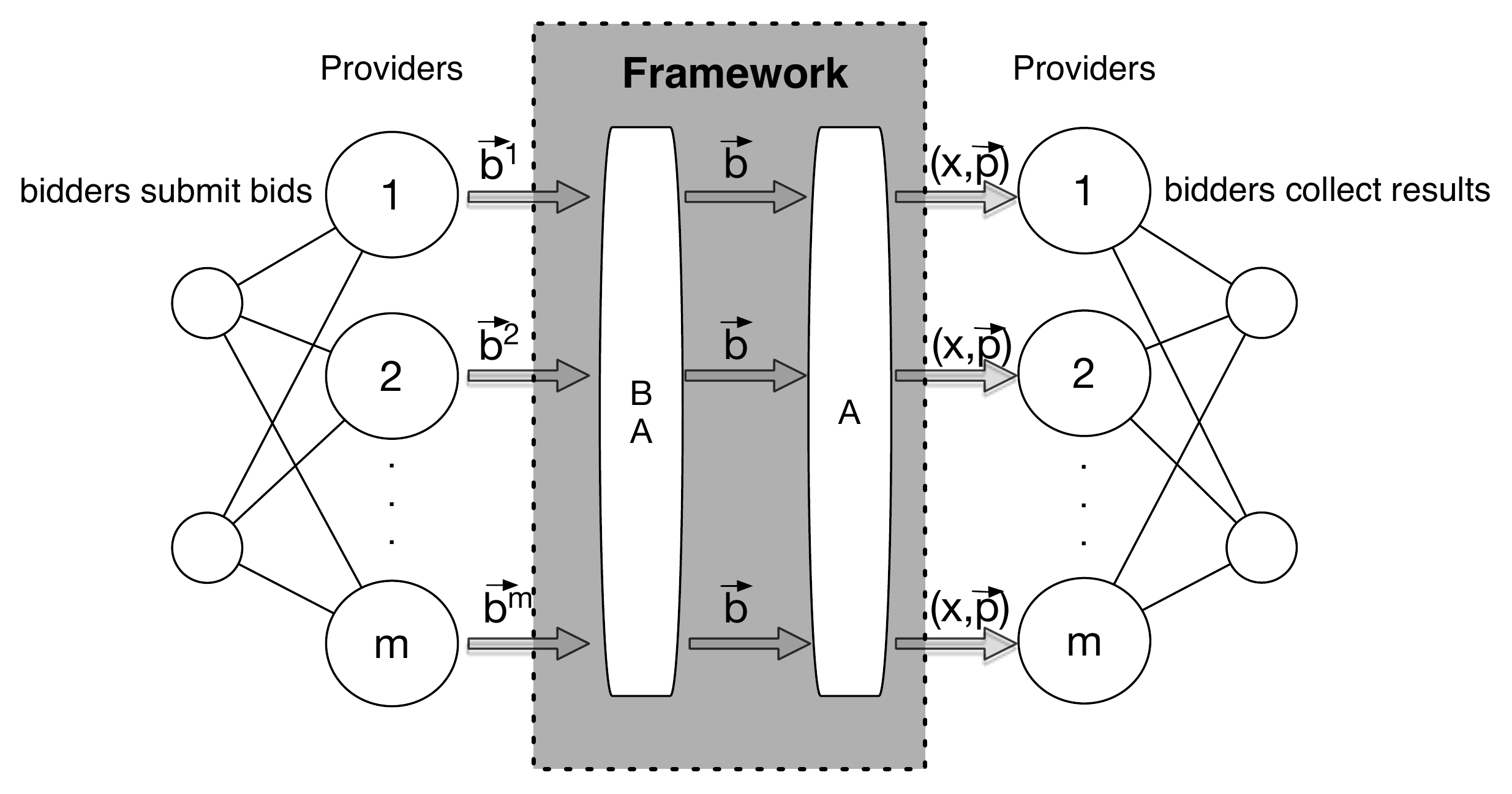}
	\caption{Framework: Bid Agreement (BA) and Allocator (A)}
	\label{fig:framework}
\end{figure}

In the following paragraphs, we describe each block 
in more detail by defining properties that must be satisfied 
by any implementation of the block, and then show in the analysis that
every implementation of the framework is $k$-resilient and correctly simulates
the auctioneer based only on the properties of the blocks. This makes
the proof independent from the actual implementation.
In all blocks, an implementation $P$ must satisfy
the property of \emph{$k$-resiliency for solution preference}, i.e.,
$P$ must be a $k$-resilient equilibrium, under
the assumption that players have preference for a solution 
and number of agents not in the same coalition is sufficiently high.
Specifically, the output of every block is either some valid value or $\bot$.
We can split the set of outcomes of the block (combinations of outputs)
into the set $A$ of solutions where all providers output the same valid value
and the set $B$ of remaining outcomes. 
In a correct execution, we want the outcome to lie in $A$.
To ensure this and that the protocol is a $k$-resilient equilibrium,
we need to assume that providers obtain a higher utility for outcomes in $A$
than for outcomes in $B$ (preference for a solution), and $m > f(k)$ for
some function $f$ defined for every $k > 0$.
The assumption of preference for a solution of the framework
is equivalent to providers preferring to receive the payments.

\subsubsection*{Bid Agreement}
The input at provider $j$ is the vector $\vec{b}^j$ of bids sent to $j$.
The output is a vector $\vec{b}$ or the special value $\bot$.
In addition to $k$-resiliency for solution preference,
this block must ensure two conditions when all providers follow the protocol: 
(1) \emph{eventual agreement}, defined as
all providers eventually outputting the same vector $\vec{b}$,
and (2) \emph{validity}, defined as,
for every bidder $i$ that submits the same bid $b_i'$
to all providers, the output at every provider is $b_i = b_i'$.

\begin{property}
\label{prop:bid-agreement}
A protocol $P$ implements bid agreement if and only if it satisfies two conditions:
(1) if all providers follow $P$, then $P$ satisfies eventual agreement and validity;
and (2) $k$-resiliency for solution preference.
\end{property}

If we can assume that the bids of malicious bidders are obtained from a finite
set of values and are equally likely, then 
a suitable approach is to use the rational consensus protocol proposed in~\cite{Afek:14},
which has inputs $\{0,1\}$ and outputs in $\{0,1,\bot\}$,
and satisfies the following two properties:
(a) if all providers follow the protocol, then all providers eventually
output the same bit, which is input by some provider;
and (b) $k$-resiliency for solution preference, assuming $m > 2k$
and that the input of every provider not in the same coalition
is either the same value or is $0$ or $1$ with equal probability.
This protocol can be used to implement the bid agreement as follows.
For each bidder $i$, provider $j$ generates a stream of bits uniquely
determined from $b_i^j$ and inputs each bit to a rational consensus instance;
if some instance outputs $\bot$, then $j$ outputs $\bot$,
otherwise, $j$ converts the stream to a bid $b_i$ and outputs a bid $b_i^*$,
where $b_i^* = b_i$ if $b_i$ is valid, or $b_i^*$ is some pre-determined valid bid otherwise.
To distinguish between different instances of rational consensus,
providers may append to the messages of each instance
the identifier of each bidder and the position of each bit.
Clearly, providers only output valid bids or the value $\bot$.
By (a), if all providers follow the protocol, then
eventual agreement and validity hold, showing (1).
Condition~(2) follows directly from (b) and $m>2k$ if
the input of every provider satisfies the assumptions of (b).
To see why these assumptions are true, notice
that, for each bidder $i$, if $i$ is not malicious, 
the input of all providers not in the same coalition is $i$'s true bid,
and if $i$ is malicious, then the bid $b_i^j$ sent by $i$ to $j$ is uniformly
distributed. If the set of possible bids is the set of all integers, 
then the stream of bits obtained from $b_i^j$ is also random.
These are reasonable assumptions, 
since we expect the behavior of malicious bidders to be arbitrary.

\subsubsection*{Allocator} The input at every provider is a vector $\vec{b}$ of bids,
and the output is either a pair $(x,\vec{p})$ or $\bot$.
We want the allocator to satisfy four conditions. First, we want the allocator
to correctly simulate $\A$, i.e., given that
all providers input the same vector $\vec{b}$ and follow the protocol,
every provider must eventually output pair $(x,\vec{p})$
with probability $\A(x,\vec{p} \mid \vec{b})$.
Second, we want resilience to collusive influences,
defined as, for all coalitions $K$ of at most $k$ elements,
if all providers not in $K$ input $\vec{b}$
and follow the protocol, then no $j \notin K$ outputs a pair $(x,\vec{p})$
with probability higher than $\A(x,\vec{p} \mid \vec{b})$,
regardless of the protocol followed by providers in $K$.
Intuitively, no coalition $K$ can influence the output of providers not in $K$,
except that they may output $\bot$ with higher probability.
Third, we want input validation to ensure that
providers have preference for solutions at the bid agreement.
More precisely, if two providers input different vectors
and follow the protocol, then they both output $\bot$,
regardless of the protocol followed by other providers.
Finally, we want $k$-resilience for solution preference
given that all providers have the same input.

\begin{property}
\label{prop:allocator}
A protocol $P$ implements the allocator if and only if it satisfies four conditions:
(1) correct simulation of $\A$; (2) resilience to collusive influence;
(3) input validation; and (4) $k$-resiliency for solution preference
if all providers have the same input.
\end{property}

We discuss implementations of the allocator in Section~\ref{sec:implementation}.

\subsubsection*{Analysis} We show in Theorem~\ref{theorem:simulation}
that a protocol that implements our framework correctly
simulates the auctioneer and is $k$-resilient.

\begin{theorem}
\label{theorem:simulation}
For every protocol $P$ that implements the framework,
$P$ correctly simulates the auctioneer,
and there exists a function $f$ such that, if $m > f(k)$,
then $P$ is a $k$-resilient equilibrium.
\end{theorem}

\begin{proof}
First, we show that $P$ correctly simulates the auctioneer.
Every provider $j$ inputs $\vec{b}^j$ to the bid agreement.
By (1) of Property~\ref{prop:bid-agreement}, regardless of the inputs,
all providers output the same vector $\vec{b}$
that satisfies validity. By (1) of Property~\ref{prop:allocator},
the outcome of the simulation is pair $(x,\vec{p})$
with probability $\A(x,\vec{p} \mid \vec{b})$.
This concludes the first step of the proof.

Now, we show that $P$ is a $k$-resilient equilibrium for $m > f(k)$ for some $f$.
Fix a coalition $K$. We take $f$ to be larger for all $k$ than the minimum value of $m$
required by Properties~\ref{prop:bid-agreement} and~\ref{prop:allocator}.
These properties imply that, if providers have preference for a solution
at the bid agreement, then the implementations of bid agreement is $k$-resilient,
so providers in $K$ prefer to follow $P$ for bid agreement.
Since this guarantees that all providers have the same input at the
allocator, the implementation of the allocator is also $k$-resilient,
implying that $P$ is $k$-resilient.

Now, we show that solution preference holds for both blocks.
Recall that the outcome is not $\bot$ only if
providers not in $K$ output the same pair,
and, if the outcome is $\bot$, then the utility is $0$.
Hence, providers in $K$ prefer (obtain an expected utility at least as high)
that providers not in $K$ output the same pair $(x,\vec{p})$;
in this case, they clearly prefer to output $(x,\vec{p})$ as well, thus they have preference
for a solution at the allocator. Now, consider the bid agreement.
The utility of an outcome of this block is the expected
utility given that providers not in $K$ follow $P$ and providers in $K$
follow an arbitrary protocol. Clearly, providers in $K$
prefer that no provider in $K$ outputs $\bot$.
By (3) of Property~\ref{prop:allocator}, providers in $K$
prefer that all providers not in $K$ output the same vector.
By (2) of Property~\ref{prop:allocator},
providers in $K$ cannot increase the probability of any outcome of the framework 
other than $\bot$ by deviating, thus, they cannot increase their expected utility by outputting
a vector $\vec{b}' \neq \vec{b}$ at the bid agreement.
This shows that providers have preference for a 
solution at the bid agreement, concluding the proof.
\end{proof}

\subsection{Parallel Allocator Framework}
\label{sec:implementation}

We describe a framework for implementations of
the allocator that satisfy Property~\ref{prop:allocator}.
We explore the possibility of parallelising the execution of $\A$ in multiple providers.
Although this approach introduces the overhead of communication between providers,
since $\A$ is often computationally intensive,
its parallelisation compensates for this overhead.

The framework consists in an initial invocation of a 
building block for \emph{input validation} followed
by the simulation of $\A$, which invokes two additional building blocks:
\emph{data transfer} and \emph{common coin}.
The input is a vector of bids and the output is either $\bot$ or a pair $(x,\vec{p})$.
At the invocation of each block, providers either output a valid value
or $\bot$; in the latter case, they output $\bot$ at the allocator.
To describe the simulation of $\A$,
it is useful to characterise the execution of $\A$
in terms of a graph of tasks, where nodes correspond to
tasks to be executed in sequence and edges 
represent data dependencies. This graph establishes a partial order
of tasks; every two tasks that are not ordered can be executed in parallel
by different providers. Figure~\ref{fig:taskdecomposition}
gives an example of a graph of $4$ tasks,
where tasks T2.1 and T2.2 can be executed in parallel.
To cope with collusion, each task $T$ is assigned to a set $S$ of at least $k+1$ providers.
If a task $T'$ is to be executed by a set $O \neq S$ of providers and $T'$ depends on
the result of $T$, then the providers of $S$ transfer data to the providers of $O$
using the data transfer building block. In a correct simulation of $\A$,
there must be one final task that depends on all other tasks,
where all providers gather all the required data to produce the final output.
Whenever providers need a random number distributed according to a 
probability distribution $\Pi$, they invoke the common coin with input $\Pi$.
Figure~\ref{fig:parallel-allocator} illustrates the framework
for the task decomposition of Figure~\ref{fig:taskdecomposition}.

\begin{figure}[tbp]
	\centering
	\includegraphics[width=0.4\columnwidth,keepaspectratio]{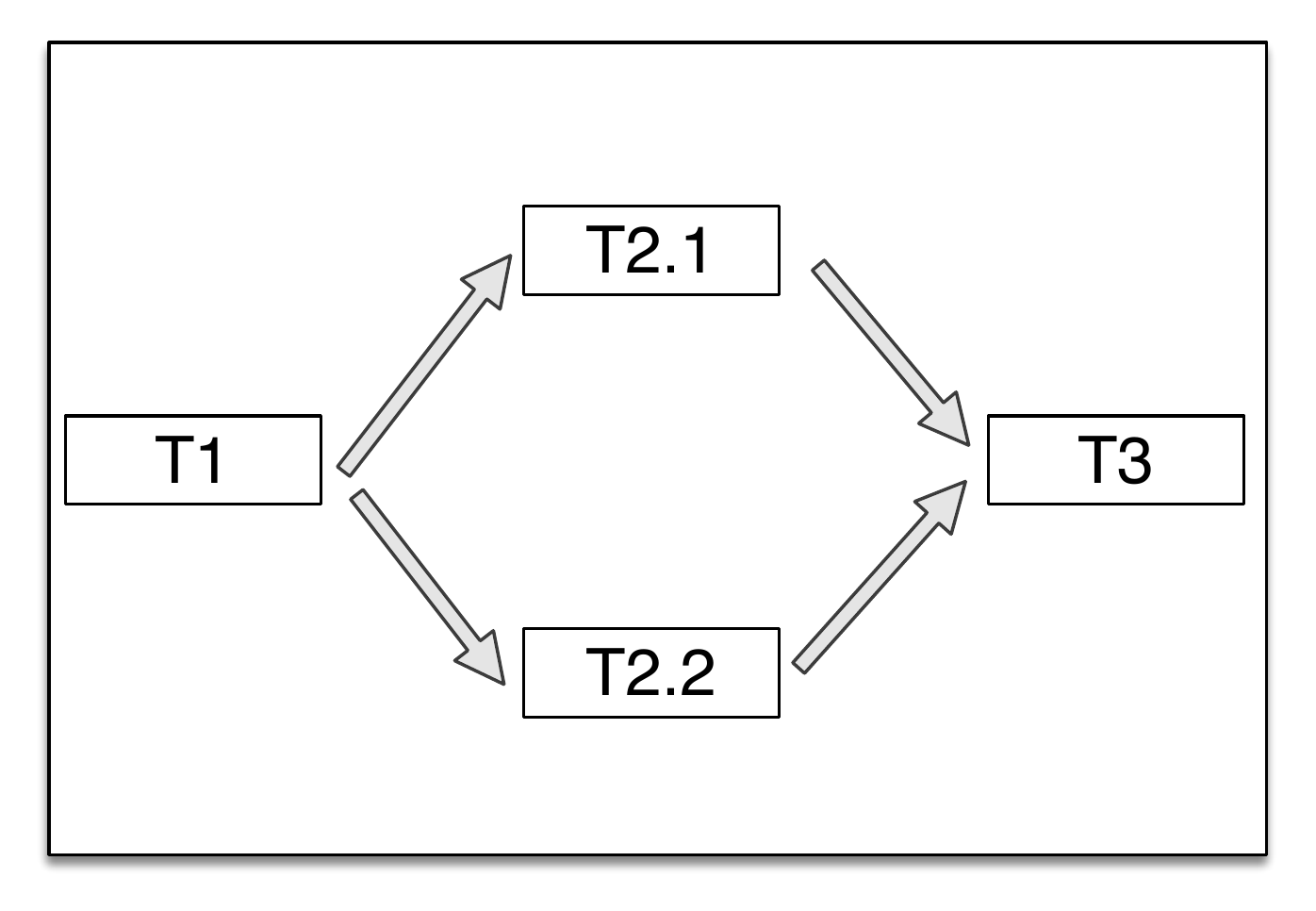}
	\caption{Decomposition of the Allocator into Tasks}
	\label{fig:taskdecomposition}
\end{figure}

As in the previous section, we describe properties that must be satisfied
by the implementations of each block and then show that every implementation
of this framework satisfies Property~\ref{prop:allocator}.

\subsubsection*{Input Validation}
The input is a vector $\vec{b}$ and the output is either
$\bot$ or $\vec{b}$. We want an implementation to satisfy 
$k$-resiliency for solution preference and that all providers 
eventually output $\vec{b}$ given that they all input $\vec{b}$,
and we need to satisfy (3) from Property~\ref{prop:allocator}.

\begin{property}
\label{prop:iv}
An implementation $P$ of the input validation must satisfy three conditions:
(1) if two providers follow $P$ and have different inputs, then they eventually output $\bot$;
(2) if all providers follow $P$ with the same input $\vec{b}$, then they eventually output $\vec{b}$;
and (3) $k$-resiliency for solution preference if all providers have the same input.
\end{property}

A simple implementation is to have providers broadcasting their vectors of bids
and outputting $\bot$ when two different vectors are detected.
This clearly satisfies (1) and (2),
whereas (3) is immediately true if providers have preference for a solution and $m > k$.

\subsubsection*{Common Coin}
The input is a probability distribution $\Pi$
and the output is either $\bot$ or a number distributed according to $\Pi$.
Given that all providers have the same input, we 
want the common coin to satisfy $k$-resilience for solution preference 
and to output the same random number.

\begin{property}
\label{prop:common-coin}
Given that all providers have input $\Pi$,
an implementation $P$ of the common coin must satisfy two conditions:
(1) if all providers follow $P$, then they eventually output 
the same value distributed according to $\Pi$;
and (2) $k$-resiliency for solution preference.
\end{property}

\begin{figure}[tbp]
	\centering
	\includegraphics[width=0.5\columnwidth,keepaspectratio]{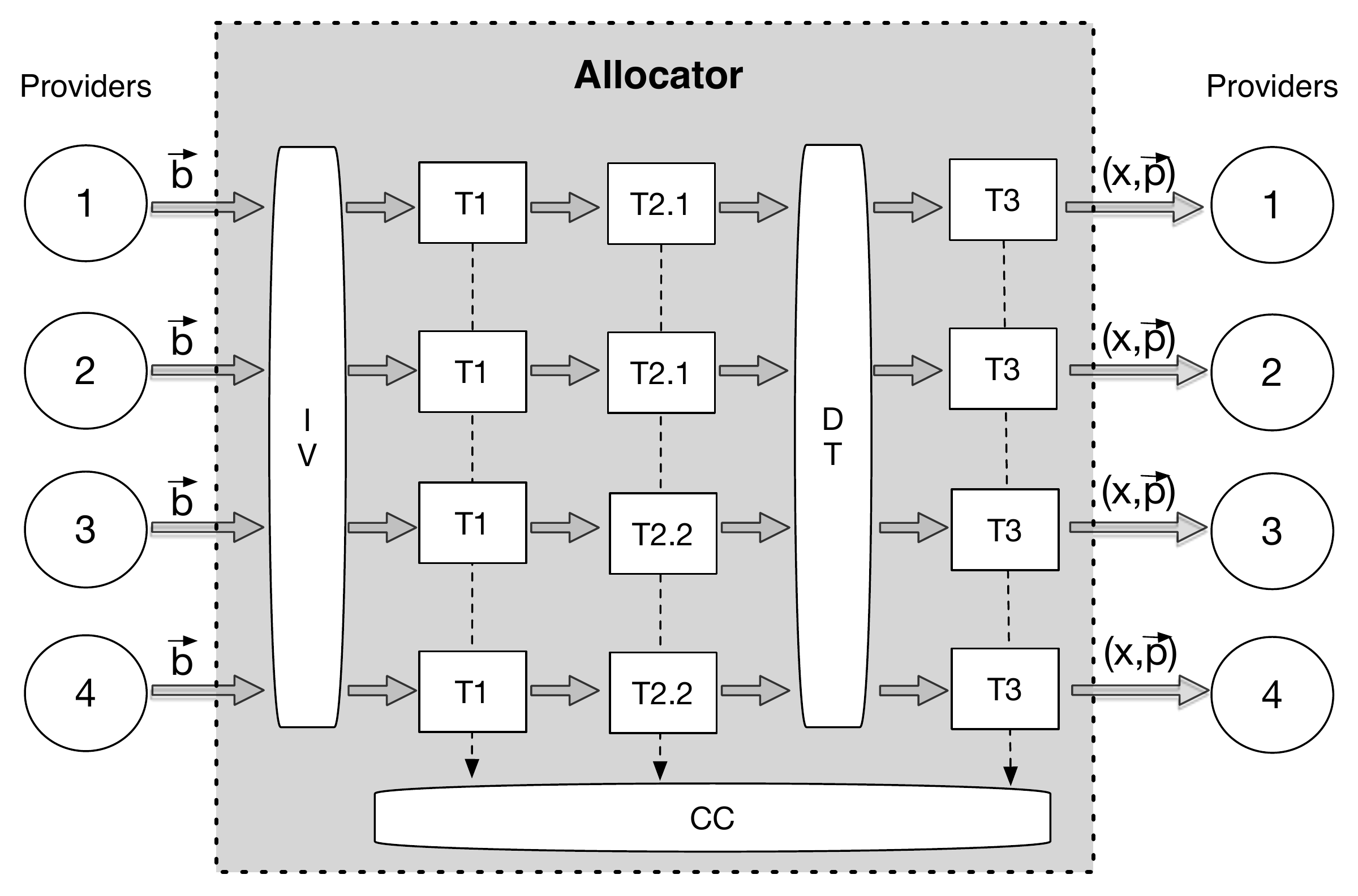}
	\caption{Parallel Allocator: Input Validation (IV), Data Transfer (DT), and Common Coin (CC)}
	\label{fig:parallel-allocator}
\end{figure}

A possible implementation of the shared coin is the protocol from~\cite{Abraham:13}.
The idea is that every provider $j$ commits to a random number $r_j \in [0,1]$,
before learning the random numbers of every other provider not in its coalition.
Then, providers reveal all random numbers and compute the output by summing all numbers modulo $1$.
If some provider $j$ sees a number not in $[0,1]$ or some provider does not send
a random value compatible with its commitment, then it outputs $\bot$.
Otherwise, $j$ applies a transformation on the computed value, which is uniformly distributed in $[0,1]$,
to produce an output that is distributed according to $\Pi$.

It is clear that all providers output the same random number distributed according to the common input $\Pi$
if they follow the protocol. Assuming that $m > k$, no provider $j$ can manipulate
the probability distribution of the output by not committing to $r_j$ selected at random
without some provider outputting $\bot$, even if $j$ is in a coalition of at most $k$ providers.
Therefore, the protocol satisfies $k$-resiliency for solution preference.

\subsubsection*{Data Transfer} A set $S$ of providers inputs a value from a domain $D$.
Providers from a set $O$ either output a value from $D$ or $\bot$.
When all providers in $S$ have the same input,
we want them to output the same value in $D$ when they follow the protocol.
We only require an implementation to be $k$-resilient if $|S|,|O| > k$,
since otherwise coalitions can always manipulate the output of this block.

\begin{property}
\label{prop:data-transfer}
Given that $|S|,|O| > k$ and all providers have the same input $v$,
an implementation $P$ of the data transfer must satisfy two conditions: 
(1) if all providers follow $P$, then they eventually output $v$;
and (2) $k$-resiliency for solution preference.
\end{property}

We propose a simple $k$-resilient implementation of this block,
where providers in $S$ broadcast their input to all providers in $O$.
In the end, if some provider $j \in O$ detects two different values,
then $j$ outputs $\bot$. Given that all providers
have input $v$ and that $|S|,|O| > k$,
they eventually output $v$, and
no coalition $K$ of up to $k$ providers
can cause all providers to output $v' \notin \{v,\bot\}$.
By solution preference, no provider in $K$ gains if someone lies
about the input $v$ or omits a message.

\subsubsection*{Analysis}
Theorem~\ref{theorem:allocator} shows that every implementation
of the above framework satisfies the four conditions of Property~\ref{prop:allocator}.

\begin{theorem}
\label{theorem:allocator}
Every protocol $P$ that implements the parallel allocator
satisfies Property~\ref{prop:allocator}.
\end{theorem}

\begin{proof}
We show that $P$ ensures (1) correct simulation of $\A$; (2) resilience to collusive influence;
(3) input validation; and (4) $k$-resiliency for solution preference if all providers have the same input.
First, we show (1). Suppose that all providers input the same vector $\vec{b}$
and follow $P$. We show that every provider outputs the same pair $(x,\vec{p})$
with probability $\A(x,\vec{p} \mid \vec{b})$. We show using induction that,
if the decomposition of $\A$ into tasks is done correctly and we fix all random
numbers, then at every task $T$ every provider $j$ that executes $T$ has the same output
that he would have if $j$ executed $\A$ locally with the same random numbers.
This is true for the first task by (2) of Property~\ref{prop:iv}. In the inductive step,
the input at each task depends only on the output of a set of tasks.
For each of those tasks $T$, by the induction hypothesis,
a set $S$ of at least $k+1$ providers compute the 
same result and input it to the data transfer;
by (1) of Property~\ref{prop:data-transfer}, all providers that execute $T$
receive that value and perform the same computation as they would if they were executing $\A$.
This implies that all providers output the same pair at the end.
By (1) of Property~\ref{prop:common-coin}, at every invocation
of the common coin, all providers input the same distribution $\Pi$ and output the same random number
distributed according to $\Pi$, where $\Pi$ is specified by $\A$. This proves (1).

Now, we show (2). Fix a coalition $K$ and suppose that all providers 
not in $K$ follow $P$ with input $\vec{b}$ and providers in $K$ follow an arbitrary $P' \ne P$.
The only way that providers in $K$ could cause providers not in $K$
to return pair $(x,\vec{p})$ with probability higher than $\A(x,\vec{p} \mid \vec{b})$
is if the result of some task used in the input of another task or as the final output
is not distributed as specified by $\A$ and $\vec{b}$.
Since each task is executed by more than $k$ providers, 
using an identical reasoning to the proof of (1), we can show using induction
that providers in $K$ cannot manipulate the probability distribution over the results
of each task, except only by increasing the probability of some provider 
not in $K$ outputting $\bot$. Here, we use the fact that,
by (3) of Property~\ref{prop:iv} and
(2) of Properties~\ref{prop:common-coin},~\ref{prop:data-transfer},
providers in $K$ cannot manipulate the probability distribution over
outputs of the building blocks in a way that increases
the expected utility of some provider in $K$. This proves (2).

Condition (3) follows by (2) of Property~\ref{prop:iv}.
To show (4), we first need to prove that providers have preference for a solution at all invocations
of building blocks, assuming that they have preference for a solution of the allocator.
Fix a coalition $K$. It is clear that providers in $K$ prefer that providers not in $K$
do not output $\bot$ at all invocations. Now, we use backwards induction to show
that they prefer that providers not in $K$ never return different values.
In the last invocation, this is clearly true by preference for a solution of the allocator.
Continuing backwards, if two providers not in $K$ output different values
at the same invocation of some block, then either they output different pairs at the end or 
input different values at the following invocation of the data transfer,
which by the hypothesis is never preferable to outputting the same value at the considered invocation.
By the proof of (2), providers in $K$ cannot manipulate the final outcome
by not outputting the same values at all invocations, so
they also prefer to output the same values as providers not in $K$,
showing solution preference at all invocations. This also shows that
providers prefer to have the same input at all invocations.
Thus, given that all providers have the same input,
no provider in $K$ can increase its expected utility
if some provider $j \in K$ does not compute each task correctly.
By (3) of Property~\ref{prop:iv} and 
(2) of Properties~\ref{prop:common-coin} and~\ref{prop:data-transfer},
$P$ is a $k$-resilient equilibrium.
\end{proof}


\section{Case Study}
\label{sec:instances}

In this paper, we use community networks as a concrete example of one of the many application areas for the distributed simulation of a trusted auctioneer.

\subsection{Community Networks}
Community networks are a bottom-up social initiative that provide Internet and networking services using off-the-shelf network hardware and the open, unlicensed wireless spectrum,
with optical fibre in some cases,
based on self-servicing and self-management by the users~\cite{Selimi2015Cloud}.  
Many of such initiatives have proven quite successful, and 
	include AWMN in Greece, FunkFeuer in Austria, Freifunk in Germany, Ninux in Italy,
and Guifi.net which is one of the largest community networks in the wold and connects more than 28,000 locations.

Community networks are based on the principle of reciprocal sharing and most of their users are moved by altruistic principles. 
However, as any other human organisation, these networks are not immune to overuse, free-riding, or under-provisioning, specially in scenarios where users may have motivations to compete for scarce resources. 
In this paper, we consider the concrete problem of bandwidth reservation on the gateways that connect the community network to the Internet. 
In particular, we consider that, as in most community networks, the subset of users that offer gateway services to the Internet is smaller than the complete set of users. 
Users that are not at the gateways may be interested in reserving bandwidth in these gateways for Internet access.
If the available bandwidth at each gateway is not enough to satisfy the demand, 
one needs to implement some arbitration mechanism that optimally and fairly allocate resources, and at the same time maximises the resource utilisation and the provider's revenue. 

In the context of our model (\S~\ref{sec:model}), \emph{providers} are the owners of the gateways, and have direct access to the Internet,
while \emph{bidders} have no direct access to the Internet and rely on these gateways.
The role of \emph{auctioneer} can be taken by one of the providers, or a chosen third party,
or in the case of our proposed distributed auctioneer, simulated by some providers in a distributed fashion.
We consider \emph{resource} to be the bandwidth external to the network available at the gateways.
Note that we are not considering the bandwidth on the links internal to the community network,
because network is operated in a shared manner, 
and this bandwidth is collectively owned by the community.

\subsection{Auctions For Bandwidth Allocation}

We now show how our framework can be applied to two different bandwidth allocation problems in the context of community networks. For that purpose, we resort to two different 
algorithms that have been proposed in the literature to solve bandwidth allocation of bidders in providers. These 
algorithms rely on standard and double auctions respectively, and have different computational properties: the double auction algorithm 
provides an example of a graph with only one task that is not computationally intensive, such that decomposing its execution into 
parallel tasks does not provide a performance gain;  the standard auction algorithm provides a graph with multiple computationally 
intensive tasks that can be parallelised. Later in the paper, we will use these examples to evaluate the performance of implementations
of the framework. We use the double auctions example to measure a worst-case overhead of executing all building blocks of the framework 
compared to an execution with a centralised trusted auctioneer, and we use the standard auctions example to show that the improvements of
parallelisation can outweigh the added overhead when the execution time is dominated by computation.

\subsubsection{Double Auction}
\label{sec:instances-double-auction}

Consider an auction where each provider has a limited bandwidth to be allocated to multiple users,
and each user has a demand of bandwidth that may be satisfied by multiple providers.
Both the users and the providers declare in their bids the value given to a unit
of allocated bandwidth. An allocation gives the amount of bandwidth of each user
allocated in each provider. We want to ensure truthfulness in expectation and budget balance.
For this purpose, we use the algorithm $\A$ of~\cite{Zheng2014Star},
which provides the above properties at the expense
of social welfare. The idea is to order the providers by increasing
value and to order the users by decreasing value.
Then, users are allocated by their order to the providers using the water-filling method:
the maximum amount of bandwidth of each user is allocated
to the first available provider without exceeding its capacity,
and any unsatisfied demand of that user is allocated to the following providers using the same method.
Since the most computationally intensive task of this algorithm
is sorting, in most practical settings there is no performance gain in parallelising the execution of $\A$.
Instead, every provider executes $\A$ locally and outputs the result.
Hence, we never need to invoke the data transfer building block.

\subsubsection{Standard Auction}
\label{sec:instances-vcg}

\begin{algorithm}[tbp]
	\caption{Standard auction allocator}
	\label{alg:bandwidth-allocation-randomised}
	\small

	\begin{algorithmic}[1]
		\State Task 1: Calculate the allocation solution $x$
		\For {Each subset $S$ of bidders in parallel}
			\State Task 2.$S$: calculate payment $p_j$ of every $j \in S$
		\EndFor
		\State Task 3: Collect the outputs of each task with the data transfer and output $(x,\vec{p})$
	\end{algorithmic}

\end{algorithm}

Consider a variation of the double auction where providers
do not send bids and each bidder can only have its bandwidth demand
allocated in a single provider. Here, we aim for truthfulness
in expectation, maximal social welfare, and computational efficiency.
It is well known that a VCG mechanism can be used to provide the first two guarantees.
The difficulty is that determining the maximal
social welfare is in general an NP Hard problem,
which conflicts with the goal of computational efficiency.
To address this issue, we use the algorithm of~\cite{Zhang2015Truthful}
which adapts the VCG mechanism to achieve a tradeoff
between the two conflicting requirements. Specifically,
\cite{Zhang2015Truthful}~offers a $(1-\epsilon)$ approximation
of maximal social welfare for an arbitrarily small $\epsilon$,
while terminating in polynomial time according to smoothed analysis.

Interestingly, the randomised algorithm proposed
in~\cite{Zhang2015Truthful} has the potential for parallelisation.
In a course manner, the algorithm can be divided into three steps, depicted
in Algorithm\,\ref{alg:bandwidth-allocation-randomised}. The first
step derives an approximately optimal allocation of users to providers.
This step is hard to parallelise effectively in a distributed system,
so we run it in a single sequential task.
The second step calculates the payments for each user based on the result of the first step.
This step is computationally intensive and the payments for each user can be computed
independently and, therefore, can be easily parallelised.
The final step gathers all intermediate results to produce the output.
In our implementation, the first and third steps are executed by all providers.
In the second step, we group the providers into $c$ groups, each containing at least $k+1$
providers. Each group is assigned the computation of the payments
of a subset of $n/c$ users. Then, all providers of a group
execute the data transfer block to transfer the resulting payments to all providers.


\section{Performance Evaluation}
\label{sec:evaluation}

We have evaluated the implementations of the allocator for double and standard auctions proposed in Section~\ref{sec:instances}.
The implementations of all the remaining blocks are as suggested in Section~\ref{sec:implementation}: we use the rational consensus algorithm proposed in~\cite{Afek:14} in the implementation of the bid agreement, the input validation and data transfer blocks are implemented as simple broadcasts, and the common coin is implemented using the scheme from~\cite{Abraham:13}. 
For these implementations to be $k$-resilient equilibria, we need $m>2k$. 
This is a requirement of the rational consensus algorithm. 

Our goal is to assess the overhead of the distributed protocol, when compared to a purely centralised solution, in the case the allocation algorithm is not parallelisable, 
and to assess the potential benefits from parallelisation when computationally expensive allocation algorithms are used. For that purpose, we measure performance gains for different levels $p$ of parallelism, where $p = \lfloor m/(k+1) \rfloor$ is the maximum level of parallelism for each possible $k$ and $p=1$ represents the sequential execution by a trusted auctioneer.
We consider a fixed number of $m=8$ providers in the auctions,
but we vary the number of providers that execute each protocol.

\subsection{Hardware/Software Setup}

In order to obtain a meaningful evaluation of our approach, we have resorted to a prototype implementation on realistic hardware and software environment and deployed it in an experimental testbed for community networks, namely on nodes of the Guifi.net~\cite{ClommunityTestbed}. 
We were given access to 4 different nodes of the experimental testbed, 2 machines in Barcelona (UPC Campus), and one machine each in Barcelona (Hangar) and in Taradell, Spain.
When doing tests executed by more than 4 providers we have instantiated 
multiple VMs in each nodes, ensuring that each VM is allocated a different CPU. 
The machines are Intel Core i7-3770 3.40GHz CPU with 16 GB RAM and 1 TB hard disk, running Proxmox virtualisation engine. 
Our experiment runs in OpenVZ containers with Debian 7 x86 1 CPU, 2 GB RAM, 10 GB storage. 
We have implemented the framework in Python, using PyPy for speed reasons, and used \O{}MQ\cite{ZeroMQ} as the messaging library for the communication.

We have set up a single node that acts as a client, and generates input for all the $n$ users.
This client node sends the requests to the $m$ providers, and receives the results back from all of them.
The values for running time presented in the plots capture the time from when the inputs are generated at this client node,
till the time it receives the results from all the experiment instances.
We run the experiments for 100 rounds, and plot the average values in the graphs.

\begin{figure*}[tbp]
	\centering
	\begin{minipage}[c]{0.48\linewidth}
		\centering
		\includegraphics[width=1.0\textwidth,keepaspectratio]{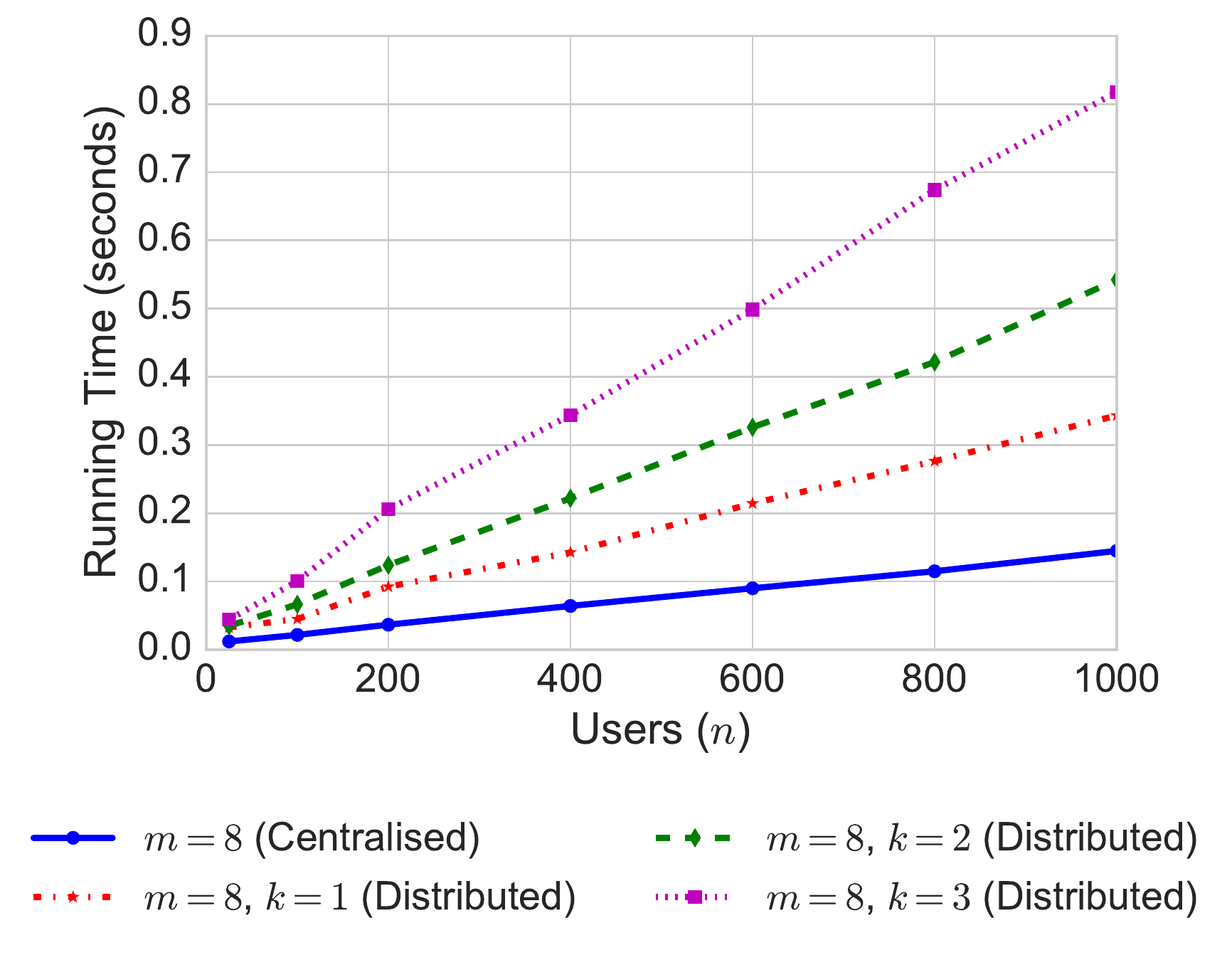}
		\caption{Running time for double auction}
		\label{fig:double-auction-time}
	\end{minipage}
	\quad
	\begin{minipage}[c]{0.48\linewidth}
		\centering
		\includegraphics[width=1.0\textwidth,keepaspectratio]{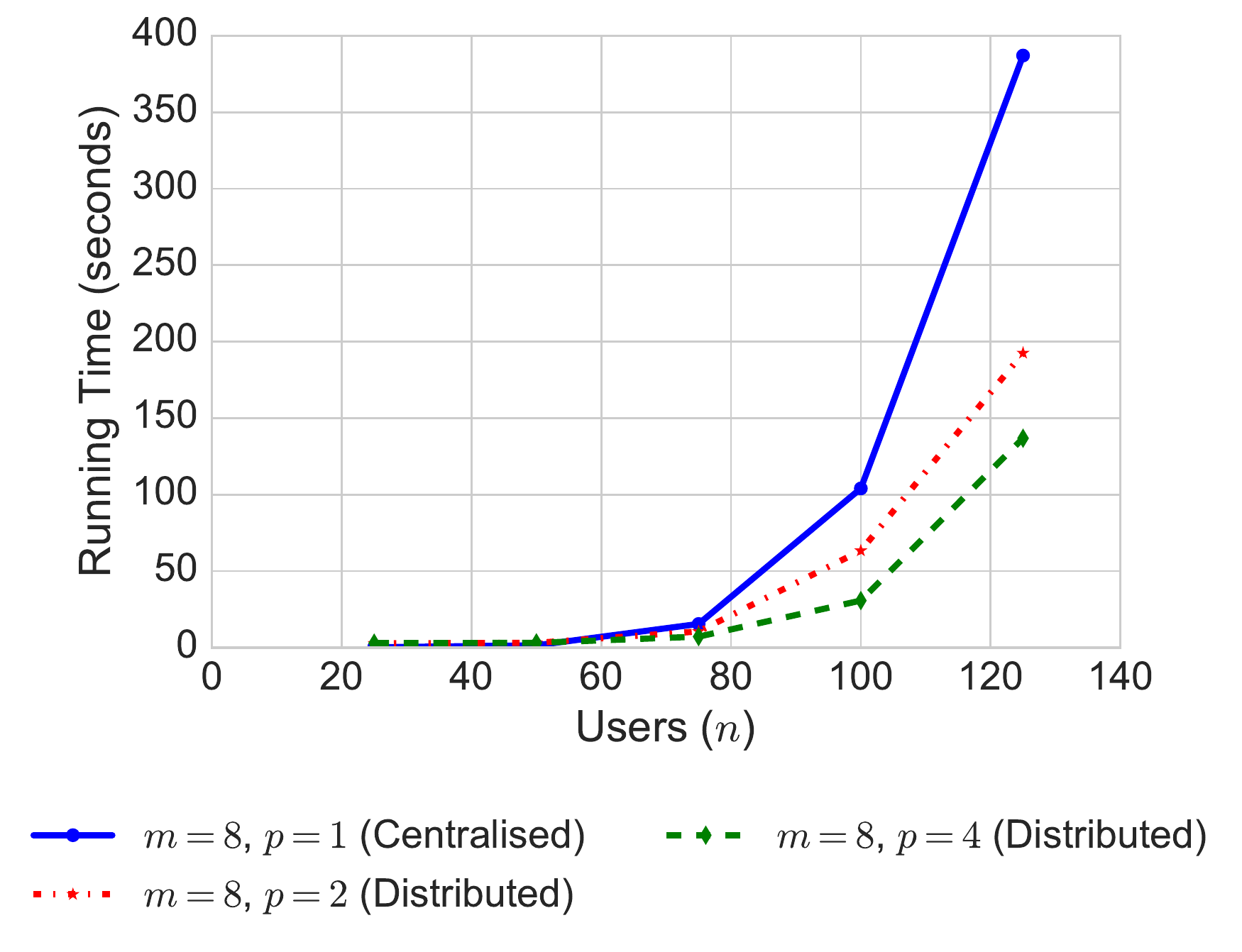}
		\caption{Running time for standard auction}
		\label{fig:standard-auction-time}
	\end{minipage}	
\end{figure*}

\subsection{Double Auction Deployment}

We have used an experimental set up similar to~\cite{Zheng2014Star}, with some slight modifications suitable to our use case.
In both the experiments, the double and standard auction, 
the bids by the users are uniformly distributed in the range $[0.75, 1.25]$,
and the requested bandwidth resource is uniformly distributed in the range $(0, 1]$.
We vary the capacity of the providers depending upon the overall bandwidth required, and scale it using a random factor in $[0.5, 1.5]$ so as to consider both the cases where providers lack the capacity to satisfy all the requests, and where the providers have excess capacity.
The providers have a unit cost of bandwidth uniformly distributed in the range $(0, 1]$.

Figure~\ref{fig:double-auction-time} shows the running time for the double auction algorithm (\S~\ref{sec:instances-double-auction}) as a function of the number of users, for up to 1000 users. 
This algorithm has little computational overhead. 
It is not easily parallelisable but, as it can be observed from the figure, this is irrelevant as the distributed version is dominated by the communication time. 
Also, the communication overhead increases as the number of users increases, since more data has to be exchanged between the providers. 
The figure shows the values obtained for the centralised approach and for the distributed implementation using different values of $k$
and corresponding minimum required number of providers out of a total of 8 involved in the execution, namely, $3$ providers when $k=1$, $5$ when $k=2$, and $8$ when $k=3$. 
Even when 8 providers and 1000 users are used, the 
distributed implementation finishes in less than a second which is perfectly 
acceptable because normally these auctions need to run with reasonable intervals between them.

\subsection{Standard Auction Deployment}

We fix the number of providers to $m=8$ and vary the maximum degree of parallelisation 
by taking $p$ to be $1$, $2$, and $4$, corresponding to a centralised execution, $k=3$, and $k=1$, respectively.
Figure~\ref{fig:standard-auction-time} shows the running time for the standard 
auction (\S~\ref{sec:instances-vcg}) as a function of the number of users, for up to 125 users.
The capacity of the providers is based on the overall bandwidth required at that provider in the bids submitted by the users, and scaled down using a random factor in $[0, 0.25]$, so roughly no more than a quarter of the users win the bids. 
For higher values of $n$, the algorithm~\cite{Zhang2015Truthful} can take in the order of hours to complete, 
which is expected as its computational complexity is $\approx \mathcal{O}(mn^9(\frac{1}{\epsilon})^2)$ for $n$ users and $m$ providers, though it provides better guarantees for social welfare than other alternatives.

Figure~\ref{fig:standard-auction-time} shows that the running time in general grows quickly as $n$ increases, and there is sharp rise in the running time for values of $n$ close to 100.
This is because the running time of the algorithm~\cite{Zhang2015Truthful} is 
a function of the feasible allocation space (which can grow exponentially in the worst case) of the resource allocation problem.
Therefore, the communication and coordination overhead is not significant when compared to the running time of the allocation algorithm. 
On the contrary, the overheads involved in distributing the inputs and aggregating the results from the providers are easily offset by the gains due to parallelisation in this case. 
In Figure~\ref{fig:standard-auction-time},  we can observe significant performance 
gains in the distributed case, for $p=2$ and $p=4$ (i.e., for $k=3$ and $k=1$ respectively). 
For instance, when 8 providers are available and $k=1$ the distributed implementation
takes around 100 seconds while the serial implementation takes around 400 seconds. 
This indicates that our approach allows for scaling the allocation algorithm, 
given that when the network grows more providers also become available.



\section{Conclusion and Future Work}
\label{sec:conclusion}

Resource allocation is a fundamental problem in networked systems and the design of auction mechanisms that can provide properties such as
truthfulness, budget balance, and maximal social welfare have been extensively studied in the literature. 
These works assume a centralised trusted auctioneer that can faithfully execute the allocation algorithm. 
Unfortunately, many networked systems of today, such as ``clouds of clouds'', edge clouds, and community networks, among others, lack a central trusted point of control. 
In this paper we have addressed the theoretical and practical challenges that need to be overcome to bridge this gap. 
More precisely, we have proposed a novel distributed framework for devising Nash equilibria 
distributed simulations of the auctioneer that are resilient to asynchrony and coalitions. 
Furthermore, our framework allows for the parallelisation of the allocation algorithm, 
leveraging the distributed nature of the simulation, which is of paramount practical importance given that, 
in many allocation algorithms, achieving maximal social welfare is computationally intensive. 
We have devised implementations of the framework in a realistic testbed of one of the largest community networks deployed today, 
and have gathered experimental evidence that the overhead of the emulation 
is not significant even in the cases the allocation algorithm cannot be parallelised, 
and brings substantial gains in the case parallelisation is possible. 
This shows that our approach can be used as a building block to implement resource allocation in decentralized networks.


\section*{Acknowledgement}

Amin Khan and and Felix Freitag have been supported by 
    the EU Horizon 2020 Framework Program project netCommons (H2020-688768), 
    and by the Universitat Politécnica de Catalunya BarcelonaTech 
    and the Spanish Government under contract TIN2013-47245-C2-1-R. 
Xavier Vila\c{c}a and Luís Rodrigues have been supported by 
    the ERC Grant Agreement number 307732 and by 
    Funda\c c\~{a}o para a Ci\^{e}ncia e Tecnologia (FCT)
through projects with references 
    PTDC/EEI-SCR/1741/2014 (Abyss) and
    UID/CEC/50021/2013.

%


\bibliographystyle{chicago} 
\bibliography{Bibliography/references} 


\end{document}